\documentclass[journal,twocolumn, 10pt]{IEEEtran}
 \usepackage[dvips]{graphicx}
\usepackage[dvipsnames]{xcolor}
\usepackage[skip=-4pt]{caption}
\usepackage{tikz}
\usepackage{pgfplots}
\usepackage{amsmath,epsfig,amssymb,algorithm,algpseudocode,amsthm,cite,url}

    \DeclareMathOperator{\tr}{tr}

\theoremstyle{plain}
\usepackage{empheq}
\newtheorem{theorem}{Theorem}
\newtheorem{assumption}{Assumption}
\usepackage{cleveref}

\newtheorem{cor}{Corollary}
\theoremstyle{definition}

\theoremstyle{remark}

\newtheoremstyle{specialcasestyle}{1mm}{1mm}{\upshape}{}{\bfseries\upshape}{.}{0mm}{}
\theoremstyle{specialcasestyle}

\newcommand{\bPhi}{\boldsymbol{\Phi}}

  \newcommand{\bT}{{\bf T}}
   
   \newcommand{\bR}{{\bf R}}
     \newcommand{\bC}{{\bf C}}
  \newcommand{\bb}{{\bf b}}
  \newcommand{\bA}{{\bf A}}
  
    \newcommand{\bD}{{\bf D}}
    
    \newcommand{\bI}{{\bf I}}
      \newcommand{\bh}{{\bf h}}
      \newcommand{\bH}{{\bf H}}

      \newcommand{\wbH}{\widehat{\bf H}}
      \newcommand{\wbHh}{\widehat{\bf H}^{\mbox{\tiny H}}}
      
      \newcommand{\wbh}{\widehat{\bf h}}
      \newcommand{\wbhh}{\widehat{\bf h}^{\mbox{\tiny H}}}

      \newcommand{\bx}{{\bf x}}

      \newcommand{\bZs}{{\bf Z}^{\mbox{\tiny S}}}
       \newcommand{\bZm}{{\bf Z}^{\mbox{\tiny M}}}

          \newcommand{\bz}{{\bf z}}
      \newcommand{\bg}{{\bf g}}

      \newcommand{\bQ}{{\bf Q}}

       \newcommand{\herm}{^{\mbox{\tiny H}}}
     \newcommand{\MRC}{^{\mbox{\tiny MRC}}}
    \newcommand{\LMMSE}{^{\mbox{\tiny S-MMSE}}}
    \newcommand{\MMMSE}{^{\mbox{\tiny M-MMSE}}}

\newcommand{\width}{0.8\linewidth}
\newcommand{\height}{6.0cm}
\usepackage{epstopdf}
\usepackage{color}
\usepackage{placeins}
\usepackage{float}
\usepackage{tabularx,colortbl}
\newcommand{\asto}{{\xrightarrow[N\to\infty]{a.s.}} \hspace{0.1cm}}

\begin{document}
      \title{
Line-of-Sight and Pilot Contamination Effects on Correlated Multi-cell Massive MIMO Systems}

\author{\authorblockN{Ikram Boukhedimi, Abla Kammoun, and Mohamed-Slim Alouini}\\
\authorblockA{Computer, Electrical, and Mathematical Sciences and Engineering (CEMSE) Division,\\
King Abdullah University of Science and Technology (KAUST),\\Thuwal, Makkah Province, Kingdom of Saudi Arabia.\\Email: \{ikram.boukhedimi, abla.kammoun, slim.alouini\}@kaust.edu.sa}}
\maketitle

\begin{abstract}
This work considers the uplink (UL) of a multi-cell massive MIMO system with $L$ cells, having each $K$ mono-antenna users communicating with an $N-$antennas base station (BS). The channel model involves Rician fading with distinct per-user Rician factors and channel correlation matrices and takes into account pilot contamination and imperfect CSI. The objective is to evaluate the performances of such systems with different single-cell and multi-cell detection methods. In the former, we investigate MRC and single-cell MMSE (S-MMSE); as to the latter, we are interested in multi-cell MMSE (M-MMSE) that was recently shown to provide unbounded rates in Rayleigh fading. The analysis is led assuming the infinite $N$ limit and yields informative closed-form approximations that are substantiated by a selection of numerical results for finite system dimensions. Accordingly, these expressions are accurate and provide relevant insights into the effects of the different system parameters on the overall performances. \end{abstract}
\section{Introduction}

{ By using a large number of antennas, and performing efficient spatial multiplexing, massive MIMO can generate scalable ergodic capacities { through fairly simple linear processing techniques}}. Consequently, it is a key-enabling technology of the upcoming 5G systems \cite{noncooperative-Marzetta2010,5G-Andrews2014,MMIMO-Lu2014}. Nevertheless, despite the substantial advantages of massive MIMO, many issues remain to be tackled. In multi-cell networks, studies have shown that interference emanating from pilot  contamination is a fundamental impairment that limits the achievable capacity \cite{MMIMO-Lu2014,pilot-Marzetta2012,HowMany-Jacob2013}. Over the past couple of years, massive MIMO has been extensively investigated in order to further enhance its resultant performances, and ultimately, overcome these limitations, (see \cite{noncooperative-Marzetta2010,Large-Wagner2012, HowMany-Jacob2013,lowcomplexity-abla2014,Coordinated-Ikram2017,Massive-Emil2018} and references therein). Particularly, the authors in \cite{Massive-Emil2018} show that when multi-cell combing schemes are employed, pilot contamination is no longer a fundamental asymptotic limitation, since the rates can in fact scale {with the number of antennas as this latter grows large}. It is important to note that most of these works consider Rayleigh fading channels which represent scattered signals. By the same token, they do not encompass the possibility of a Line-of-Sight (LoS) component which is commonly present in practical wireless propagation scenarios modeled by Rician fading.

\par Lately, the performances of massive MIMO with Rician fading channels have received a lot of interest \cite{Power-Zhang2014,Beamforming-Yue2015, asymptotic-falconet2016, Downlink-Zhao2016,Channel-Wu2017, Asymptotic-Luca2017}. Among these, few works consider the multi-cell setting with pilot contamination \cite{Downlink-Zhao2016,Channel-Wu2017,Asymptotic-Luca2017}. In \cite{Downlink-Zhao2016}, the authors examine channel estimation schemes for downlink transmissions of an uncorrelated massive MIMO system affected by pilot contamination. In the same line, \cite{Channel-Wu2017} proposes a novel technique that uses LoS estimates for channel estimation  in order to minimize the pilot contamination effects. Furthermore, a large system analysis relying on some recent random matrix theory results is led in \cite{Asymptotic-Luca2017}, where the authors analyze regularized-ZF and maximum-ratio-transmit precoders in the downlink of a massive MIMO system with uncorrelated Rician channels. \par { The impact of LoS propagation conditions on massive MIMO systems has not been sufficiently investigated, although several issues arise in this context. Particularly, establishing whether factors such as pilot contamination remain performance limiting, or the circumstances under which LoS conditions become beneficial, will contribute to make massive MIMO more resilient to such elements and improve the overall performances. Accordingly, investigating these issues constitutes the main focus of the present work.} To this end, we consider the uplink (UL) of a multi-cell massive MIMO system with a comprehensive system model that includes imperfect channel estimation, pilot contamination, and distinct per-user Rician factor and channel correlation matrix. Furthermore, we conduct a comparative analysis of different combining schemes. Specifically, we first investigate the performances when using single-cell detectors wherein only local channels are used to process the received signal. In this context, we examine the UL rates resulting from Maximum-Ratio-Combining (MRC) and Minimum-Mean-Square-Error (which we shall denote S-MMSE to refer to single-cell detection) receivers. Second, we analyze multi-cell combining where local and inter-cell channels are utilized in the design of the receiver. Accordingly, we consider multi-cell MMSE (M-MMSE), where we extend the study in \cite{Massive-Emil2018} to correlated Rician fading. Ultimately, we aim at identifying the limitations of the three detectors with respect to propagation conditions and the different forms of undergone interference. To this end, assuming the large-antenna limit, we derive closed-form approximations of the UL rate using some basic asymptotic tools such as the Law of Large Numbers (LLN). These expressions are tractable and instructive as they allow to investigate the effects of the LoS and pilot contamination on the spectral efficiency. Our study reveals that in single-cell processing, the system remains limited by pilot contamination and that the specular signals enhance the achievable rate, notwithstanding some LoS-induced interference. Furthermore, we demonstrate that this latter dissipates under favorable propagation conditions for MRC, and can be mitigated through certain designs of the S-MMSE. { As to M-MMSE, we find that the UL rate scales with the number of antennas, assuming asymptotically linearly independent correlation matrices as proposed in \cite{Massive-Emil2018}, for Rayleigh fading. Consequently, we validate, in multi-cell correlated Rician fading settings, the conclusions of \cite{Massive-Emil2018} that M-MMSE outperforms both single-cell detection techniques.} 
 
 \par The remainder of the paper is organized as follows\footnote{\it Notations: In the sequel, bold upper and lowercase characters refer to matrices and vectors, respectively. We use $(.)^{\text{H}}$ , $\tr(.)$ and $(.)^{-1}$ to denote the conjugate transpose, the trace of a matrix and the inverse operations, respectively. $log(.)$ is the natural logarithm, the $N \times N$ identity matrix is denoted $\textbf{I}_{N}$, and $\boldsymbol{0}_N$ is the $N\times N$ all zeros matrix. Plus, $\left[{\bf A} \right]_{ij}$ is the element on the $i-$th row and $j-$th column of matrix $\bf A$, and $diag\left\{ a_i \right\}_{i=1}^{N}$ is the $N\times N$ diagonal matrix with $a_i$ being its $i-$th diagonal element.  $\mathcal{B}\left(\bA,\ell,j\right)$ refers to the $\ell j-$th block of the block matrix $\bA$. Finally, the symbol $\delta_{j\ell}$ is the Kronecker delta such that $\delta_{j\ell}=1 $ if $\ell =j$, and $0$ otherwise.}. The system model, channel estimation and the detection schemes for the UL are introduced in the next Section. Then, Section III provides the theoretical analysis of the rate for each combining technique. Finally, Section IV provides a selection of numerical results that validate our theoretical findings and shed light on some interesting aspects and conclusions are drawn in Section V.

%\newpage
\section{System Model}
 We consider the UL of a TDD multi-cell multi-user system with $L$ cells having each $K$ single-antenna users communicating with an $N-$antennas BS. We denote $\bH_{j,\ell} \in \mathbb{C}^{N \times K}$ as the aggregate channel matrix between the UEs in cell $\ell$ and BS$_j$. {Throughout the paper, we focus on the performances of cell $j$}. Assuming Gaussian codebooks, we denote by  $\sqrt{{p}} \bx_j \sim\mathcal{CN}(0,p \bI_K)$  the vector of the transmitted data symbols sent by all the UEs in cell $j$ with the average power ${p}$. In this line, the received signal at BS$_j$ is given by : 
\begin{equation}
{\bf y}_j=\sqrt{{p}} \bH_{jj} \bx_j + \sqrt{{p}} \displaystyle{\sum_{\substack{\ell=1 \\ \ell \neq j}}^{L}} \bH_{j\ell} \bx_\ell + {\bf n}_j,
\end{equation}       
where ${\bf n}_j$ represents a zero-mean additive Gaussian noise with variance $\sigma^2\bI_N$. 
  \subsection{ Channel Model}
 We assume that the system operates over a block-flat fading channel with coherence time $T_c$. 
Let $\bh_{j\ell k}$ represent the channel linking the $k-$th UE in cell $\ell$ to BS$_j$. We consider correlated Rician fading for intra-cell or local channels and correlated Rayleigh fading for channels from other cells. { This is a reasonable setting for inter-cell channels, owing to the longer distances between UEs and the BSs in other cells, that would likely include scatterers and thus significantly reduce the possibility of a Line-of-Sight transmission.} Specifically, $\bh_{j\ell k}\in \mathbb{C}^{N\times 1}$ is given by:
\begin{align}
\bh_{jjk}= & \sqrt{\beta_{jjk}} \left(\sqrt{\frac{1}{1+\kappa_{jk}}}\boldsymbol{\Theta}_{jjk}^{\frac{1}{2}}\bz_{jjk} + \sqrt{\frac{\kappa_{jk}}{1+\kappa_{jk}}}\overline{\bz}_{jk}\right), \\
\bh_{j\ell k}= & \sqrt{\beta_{j\ell k}} \boldsymbol{\Theta}_{j\ell k}^{\frac{1}{2}}\bz_{j\ell k}, \quad \ell \neq j,
 \end{align}
where $\beta_{j\ell k}$  represents the large-scale fading and the remaining terms account for the small-scale fading. Moreover, the channel model includes a Rayleigh-distributed component $\bz_{j\ell k} \sim \mathcal{CN}\left(0,\bI_N\right)$ to account for the scattered signals. Plus, for the intra-cell channels, i.e. $\ell=j,$ a deterministic component $\overline{\bz}_{j k}\in\mathbb{C}^{N}$ is considered to represent the specular (LoS) signals. The Rician factor $\kappa_{jk} \geq 0$ denotes the ratio between the specular and scattered components.  Note that, for each UE $k$, we consider a different Rician factor $\kappa_{j k}$, due to the distinct geographic locations of the UEs, and a different channel correlation matrix $\boldsymbol{\Theta}_{j\ell k}$, as well. Throughout the paper, $\forall (j,\ell, k)$, $\boldsymbol{\Theta}_{j \ell k}$ is assumed to be slowly varying compared to the channel coherence time and thus is supposed to be  perfectly known to the BS. For notational convenience, let :
\begin{align}
{\bf R}_{j \ell k} ={\frac{\beta_{j \ell k} }{1+\kappa_{jk}\delta_{j\ell}}}\boldsymbol{\Theta}_{j \ell k}, 
\end{align}
and the aggregate matrix of the LoS components in cell $j$:  $\overline{\bH}_{j}=\left[\overline{\bh}_{j1} \overline{\bh}_{j2} \dots \overline{\bh}_{jK}\right]$ with: 
\begin{align}
\overline{\bh}_{jk} =\sqrt{\frac{\beta_{jjk}\kappa_{jk}}{1+\kappa_{jk}}}\ \overline{\bz}_{jk}.
\end{align}
 \subsection{Channel Estimation}
During a dedicated uplink training phase of $\tau$ symbols, the UEs in each cell transmit mutually orthogonal pilot sequences which allow the BSs to compute estimates of the channels. As shall be discussed next, depending on the type of receiver in use, BS$_j$ either solely estimates its local channels $\bh_{jjk}$, $\forall k$, or all the channels linking it to UEs,\textit{ i.e.} $\bh_{j\ell k}$, $\forall (\ell,k)$.  %\textbf{Why ${\bf y}^{tr}_{jk}$ ?}
In addition, since it is a multi-cell system, we assume that the same set of pilot sequences is reused in every cell so that the channel estimates are corrupted by pilot contamination from adjacent cells. In fact, the same pilot is assigned to every $k$-th UE in each cell, and as such $\forall (j,k)$, the estimates of the channels $\bh_{j 1 k}, \bh_{j2k},\dots,\bh_{jLk}$, will be correlated. Accordingly, using the MMSE estimation, the estimate of $\bh_{j\ell k}$ is given by\cite{book-Kay97,Power-Zhang2014}:  
 \begin{equation}
{\wbh}_{j\ell k}= {\bf R}_{j\ell k } \bPhi_{jk} \left( \sum_{\substack{\ell'=1 }}^{L} \bh_{j \ell' k } + \frac{1}{\sqrt{\tau \rho_{tr}}} {\bf n}_{jk}^{tr}\right)+ \delta_{j \ell} \overline{\bh}_{jk},
\label{eq:h_hat M-MMSE}
\end{equation}
where
%\begin{equation}
$\bPhi_{jk} =\left(\sum_{\ell' =1}^{L}{\bf R}_{j \ell' k} + \frac{1}{\tau \rho_{tr}}\bI_N\right)^{-1}$ ,
%\end{equation}
with $\rho_{tr} ={\dfrac{p}{\sigma^2}}$ and ${\bf n}_{jk}^{tr}$ are, respectively, the SNR and noise during training. In fact, the higher value $\rho_{tr}$ takes, the better quality of channel estimation becomes. Therefore, it can be easily shown that ${\wbh}_{j\ell k}$ $\sim \mathcal{CN}\left(\delta_{j \ell} \overline{\bh}_{jk}, \tilde{\bR}_{j\ell k}\right)$, with $\tilde{\bf R}_{j\ell k}={\bf R}_{j\ell k} \bPhi_{j k}{\bf R}_{j\ell k}$. Plus, due to MMSE estimation orthogonality, the estimation error $\boldsymbol{\xi}_{j\ell k}=\bh_{j\ell k}-\wbh_{j\ell k} $, follows the distribution $\boldsymbol{\xi}_{j\ell k}\sim \mathcal{CN}\left(0,\bR_{j\ell k}-\tilde{\bR}_{j\ell k}\right)$. 
%{\color{red} In MRC and S-MMSE only local channels are estimated and used to process the received signals.}

\subsection{Detection and Uplink Rates}
In order to retrieve the signal sent by UE $k$, BS$_j$ uses a linear receiver $\bg_{jk}\in \mathbb{C}^{N\times 1}$ to process the received signal ${\bf y}_j$ by generating ${r}_{jk}={\bg}_{jk}\herm{\bf y}_j$. Therefore:
\begin{align}
&{r}_{jk}=  \overbrace{\sqrt{{p}} \bg\herm_{jk} \wbh_{jjk} x_{jk}}^{\text{signal}} + \overbrace{\sqrt{{p}}  \bg\herm_{j k} \boldsymbol{\xi}_{jk} x_{jk} }^{\text{estimation error}} \notag \\ & 
+ \underbrace{\sqrt{{p}}{\sum_{\ell=1}^{L}\sum_{\substack{i=1\\ i \neq k}}^{K} \bg\herm_{j k} \bh_{j\ell i} x_{\ell i}}}_{\text {intra+inter-cell interference}} + \underbrace{\sqrt{{p}}{\sum_{\substack{\ell=1 \\ \ell\neq j}}^{L} \bg\herm_{j k} \bh_{j\ell k} x_{\ell k}}}_{\text{pilot contamination}}+  \underbrace{\bg\herm_{jk} n_{jk}}_{\text{processed noise}}.%\bg\herm_{jk} n_{jk}
\label{eq:decomposed r}
\end{align}
%This expression respectively, separates the signal, channel estimation error, intra and inter-cell interference, pilot contamination and noise terms.
Let $\rho_d = \frac{p}{\sigma^2}$, denote the SNR during data transmission. Therefore, the SINR corresponding to the transmission between the $k$-th UE in cell $j$ and its BS, $\gamma_{jk}$, is defined as:
{\begin{align}\label{eq:gamma def}
&\gamma_{jk}=\notag \\ 
&\frac{\vert{\bf g}_{jk}\herm\wbh_{jjk}\vert^2}
{\mathbb{E}\left[{\bf g}_{jk}\herm\left(\displaystyle{\sum_{\substack{(\ell,i)\neq (j, k)}} \bh_{j\ell i}\bh_{j\ell i}\herm } + \left.\boldsymbol{\xi}_{jk}\boldsymbol{\xi}_{jk}\herm+\frac{1}{\rho_d}\bI_N\right){\bf g}_{jk}\right|\wbH_{j}\right]}.
\end{align}}
Additionally, the achievable UL rate at BS$_j$ is defined as: 
\begin{align}\label{eq:sum-rate}
{R}_j=\left(1-\frac{\tau}{T_c}\right) \frac{1}{K} \sum_{k=1}^{K}\mathbb{E} \left[\log\left(1+\gamma_{jk}\right)\right].
\end{align}

\subsubsection*{Signal Detection} 
We propose to investigate three designs for $\bg_{jk}$. First, we consider the case wherein the BS uses only its legitimate channels to perform the signal detection. In this line, two linear receivers are of practical interest in massive MIMO systems, namely the matched filter or MRC and the MMSE detector. Previous works have shown that these receivers present asymptotically the same performance over Rayleigh fading channels. It is thus of interest to investigate whether this result still holds true in the presence of LoS components.  Since only local channels are used,  single-cell channel estimation is sufficient to build the detectors. In this context, note that we shall refer to MMSE detection in single-cell estimation by ``S-MMSE''. Hence \footnote{For sake of clarity, in the sequel, the superscripts $(.\MRC)$, $(.\LMMSE)$ and $(.\MMMSE)$ will be added to denote the corresponding quantity when using the MRC, S-MMSE and M-MMSE receivers, respectively.}:   
\begin{align}
&{\bg}\MRC_{jk} = \wbh_{jjk}, \label{eq:G_MRC_hat}\\
&{\bg}\LMMSE_{jk} =  \left(\sum_{i=1}^{K}\wbh_{jji}\wbhh_{jji} +(\bZs_{j})^{-1} \right)^{-1}\wbh_{jjk}.\label{eq:G_MMSE_hat}
\end{align}
with $(\bZs_j)^{-1}= \frac{1}{\rho_d} \bI_N+ \bA$, where $\bA$ $\in\mathbb{C}^{N\times N}$ is an arbitrary hermitian positive semi-definite design parameter. For instance, it could contain the covariances of estimation errors and inter-cell interference as in \cite{HowMany-Jacob2013}, \textit{i.e.}, we can put:  $
\bA= \sum_{i=1}^{K}({\bf R}_{jji}-\tilde{\bf R}_{jji})+\sum_{\substack{\ell=1 \\ \ell\neq j}}^{L}\sum_{i=1}^{K}{\bf R}_{j \ell i}$.
\par Second, we consider the setting where the BS further exploits the observation received during training in order to estimate the interfering channels from other cells, and ultimately, uses them to design the MMSE detector. {The authors in \cite{Massive-Emil2018} } refer to this technique as multi-cell-MMSE (M-MMSE) combining to differentiate from the above mentioned S-MMSE detector \eqref{eq:G_MMSE_hat} obtained thorough single-cell estimation. Using M-MMSE implies that BS$_j$ computes all the channel estimates $\wbh_{j\ell k}$, $\forall (\ell,k)$. Accordingly, based on the independence between the estimates and their associated estimation errors $\boldsymbol{\xi}_{j\ell k}$, the SINR expression \eqref{eq:gamma def} can be simply rewritten as: 
{%\small 
\begin{align}
\gamma\MMMSE_{jk}= \frac{\vert{\bf g}_{jk}\herm\wbh_{jjk}\vert^2}
{{\bf g}_{jk}\herm\left(\displaystyle{\sum_{\substack{(\ell,i)\neq (j, k)}} \wbh_{j\ell i}\wbh_{j\ell i}\herm } + (\bZm_j)^{-1} \right){\bf g}_{jk}},
\label{eq:SINR M-MMSE}
\end{align}}
with :  $(\bZm_j)^{-1}= \frac{1}{\rho_d} \bI_N+ \displaystyle{\sum_{\substack{\ell = 1}}^{L}\sum_{i=1}^K}({\bf R}_{j\ell i}-\tilde{\bf R}_{j\ell i})$. In fact M-MMSE is defined as the optimal receiver that maximizes the UL SINR in multi-cell estimation: 
\begin{align}
{\bg}\MMMSE_{jk}= \underset{\bg_{jk}} {\rm{argmax}} \  \gamma\MMMSE_{jk}. \label{eq:G_MMMSE_hat}
\end{align}
Clearly, finding \eqref{eq:G_MMMSE_hat} amounts to maximizing the Rayleigh-Quotient, whose well-known solution can be:
\begin{equation}\label{eq:g_MMMSE_hat}
{\bg}\MMMSE_{jk} =  \left(\sum_{\ell=1}^{L}\sum_{i=1}^{K}\wbh_{j\ell i}\wbhh_{j\ell i} +( \bZm_{j})^{-1} \right)^{-1} \wbh_{jjk}.
\end{equation}
Plugging $\bg\MMMSE_{jk}$ \eqref{eq:g_MMMSE_hat} in the SINR expression $\gamma\MMMSE_{jk}$, thus yields: 
\begin{align}\label{eq:gamma M-MMSE}
\hspace*{-0.5em}\gamma\MMMSE_{jk} = \wbhh_{jjk}\left(\displaystyle{\sum_{\substack{(\ell,i)\neq (j, k)}}}\wbh_{j\ell i}\wbhh_{j\ell i}+(\bZm_j)^{-1}\right)^{-1}\wbh_{jjk}.\hspace*{-0.3em}
\end{align} 

\section{Analysis of Achievable UL Rates}\label{sec:Analysis}
We carry out a theoretical analysis of the UL performances when using MRC, S-MMSE and M-MMSE detectors under the assumption of correlated Rician channels with imperfect CSI. To this end, we first derive closed-form expressions for the achievable uplink rates in the large-antenna limit. The obtained approximations are tractable, and as such, enable to apprehend { the trends of the UL performances with respect to the different system parameters such as the Rician factors or the correlation matrices. Additionally,  they allow to identify the factors that limit the performances in terms of spectral efficiency. Note that our results are tight and apply for systems with  different structures of  channel correlation matrices, as shall be illustrated with simulations.} Our main results are stated in three theorems treating each a different receiver. In every theorem, we provide approximations of the SINR $\gamma_{jk}$ which, according to the continuous mapping theorem\cite{book-Prob95}, yields an approximation of the ergodic achievable rate. 

As $N\rightarrow \infty$, while $K$ is maintained fixed, we consider the following assumption:
\begin{assumption}\label{ass:asymptotic}
$\forall j,\ \ell, \ k$, $\limsup_{N} \Vert\bR_{j\ell k}\Vert_2 < \infty$, and $\liminf_N \frac{1}{N}\tr(\bR_{j\ell k}) > 0$.
\end{assumption}
%In order to impose the presence of the LoS
In the sequel, we shall refer to this regime as $N \rightarrow \infty$.
\subsection{MRC}

\begin{theorem}\label{th:MRC} Under Assumption \ref{ass:asymptotic}, $\gamma_{jk}\MRC - \overline{\gamma}_{jk}\MRC \asto 0$, with: 
%Specifically when $N \rightarrow \infty$ with fixed $K$ : 
\begin{align}
&\overline{\gamma}_{jk}\MRC = \frac{ \left(\frac{1}{N}\tr\tilde{\bf R}_{jjk}+\frac{1}{N}\Vert\overline{\bh}_{jk}\Vert^2\right)^2}
{\underbrace{\frac{1}{N^2} \displaystyle{\sum_{\substack{i=1\\ i\neq k}}^{K}}\vert\overline{\bh}\herm_{jk}\overline{\bh}_{ji}\vert^2}_\text{LoS Intra-cell interference }+ \underbrace{\sum_{\substack{\ell=1 \\ \ell\neq j}}^{L}\left(\frac{1}{N}\tr(\bR_{jjk}\bPhi_{jk}{\bR}_{j\ell k})\right)^2}_{\text{pilot contamination}}}. \label{eq:gamma MRC DE}
\end{align}
\end{theorem}
\begin{proof}
A sketch of the proof is given in Appendix \ref{app:MRC and S-MMSE sketch}
\end{proof}

First, note that if only Rayleigh fading was considered, \textit{i.e.} $\kappa_{jk}=0$,$\forall(j,k)$, the approximation $\overline{\gamma}\MRC_{jk}$ coincides with the findings in \cite{noncooperative-Marzetta2010, HowMany-Jacob2013}. Second, these works state that massive MIMO systems using MRC detection are only limited by pilot contamination. In fact, regarding inter-cell interference, this is still true in Rician fading, since we can see from $\overline{\gamma}\MRC_{jk}$, that only pilot contamination remains out of the total interference inflicted by other cells. Nevertheless, in contrast to the Rayleigh fading, \eqref{eq:gamma MRC DE} reveals that in Rician fading, the system undergoes LoS induced intra-cell interference as well. This latter is characterized by the inner product  $\frac{1}{N}\overline{\bh}\herm_{jk}\overline{\bh}_{ji}$, $i\neq k$ and, thus would dissipate under asymptotic favorable propagation conditions wherein: $\frac{1}{N}\overline{\bh}\herm_{jk}\overline{\bh}_{ji}\asto 0$, $\forall i\neq k$. Therefore, in such environments, better performances are attained:
 \begin{cor}[$\gamma\MRC$ under favorable propagation]\label{cor:MRC fav} Under Assumption \ref{ass:asymptotic}, if $\frac{1}{N}\overline{\bh}\herm_{jk}\overline{\bh}_{ji}\asto 0$, $\forall i\neq k$: 
\begin{align}
&\overline{\gamma}_{jk}\MRC = \frac{ \left(\frac{1}{N}\tr\tilde{\bf R}_{jjk}+\frac{1}{N}\Vert\overline{\bh}_{jk}\Vert^2\right)^2}
{ \sum
_{\substack{\ell=1 \\ \ell\neq j}}^{L}\left(\frac{1}{N}\tr(\bR_{jjk}\bPhi_{jk}{\bR}_{j\ell k})\right)^2}. \label{eq:gamma MRC DE fav}
\end{align}
 \end{cor}

\subsection{S-MMSE}

\begin{theorem}[S-MMSE]\label{th:S-MMSE}
 Under Assumption \ref{ass:asymptotic}, we have : $ \overline{\gamma}\LMMSE	_{jk}-{\gamma}\LMMSE	_{jk} \asto 0 $,  such that:
 
{%\color{blue}
\begin{align}
&\overline{\gamma}\LMMSE_{jk} = 
&\frac{1}{
\underbrace{\displaystyle{\sum_{\substack{\ell=1 \\ \ell \neq j}}^{L}}\left|\left[\bQ_j\right]_{kk}\beta_{jk,\ell j}^{\mbox{\tiny S}}\right|^2}_{\text{induced by pilot contamination}}+\underbrace{\displaystyle{\sum_{\substack{\ell=1 \\ \ell \neq j}}^{L}\sum_{\substack{i=1\\ i\neq k}}^{K}}\left|\left[\bQ_j\right]_{ki}\beta_{ji,\ell j}^{\mbox{\tiny S}}\right|^2}_{\text{uncorrelated inter-cell interference}}}. \label{eq:gamma MMSE DE simplified}
\end{align}}
with: $\beta_{jk, \ell j}^{\mbox{\tiny S}}= \frac{1}{N}\tr(\bR_{j\ell k}\bPhi_{jk}{\bR}_{jj k}\bZs_j)$, and
\begin{align}
\bQ_j = \left(\frac{1}{N}\overline{\bH}_j\herm\bZs_j\overline{\bH}_j+ diag\left\{ \beta_{ji,jj}^{\mbox{\tiny S}}\right\}_{i=1}^{K}\right)^{-1}. \label{eq:Q_hat} 
\end{align}
\end{theorem}
\begin{proof}
A sketch of the proof is given in Appendix \ref{app:MRC and S-MMSE sketch}
\end{proof}

{ Let us examine the expression $\overline{\gamma}\LMMSE_{jk}$ \eqref{eq:gamma MMSE DE simplified}. %Similarly to MRC, the performance of S-MMSE is limited by pilot contamination given by the first term of the denominator. 
In contrast to MRC, intra-cell interference is completely eliminated in S-MMSE, and this for any propagation conditions. Furthermore, as can be seen by the the denominator of \eqref{eq:gamma MMSE DE simplified}, the entirety of inter-cell interference, including pilot contamination, hinders S-MMSE from achieving higher UL rates. These results can be explained by the structure of the S-MMSE detector, $\overline{\bg}\LMMSE$ \eqref{eq:G_MMSE_hat} that includes all the channels $\wbh_{jji}$, and thus, cancels intra-cell interference, yet, does not suppress inter-cell interference. Nevertheless, note that we can alleviate the effects of this latter by mitigating the uncorrelated inter-cell interference \footnote{Uncorrelated inter-cell interference refers to the total inter-cell interference minus the pilot contamination induced interference.} which is depicted by the second term of the denominator.  Indeed, examining this term unveils that uncorrelated inter-cell interference vanishes asymptotically when the matrix $\bQ_j$\eqref{eq:Q_hat} is diagonal. We propose to achieve this by a proper selection of $\bZs_j$ that designedly eliminates the off-diagonal elements of $\bQ_j$. Accordingly, we suggest two structures of $\bZs_j$ based on the propagation conditions. First, in favorable propagation conditions, it is sufficient to put $\bZs_j=\rho_d\bI_N$. Note that in such a setting, S-MMSE and MRC deliver the same performances given in corollary \ref{cor:MRC fav}. Second, in non-favorable propagation, \textit{i.e.} $\liminf_N \frac{1}{N}\overline{\bh}_{ji}\herm\overline{\bh}_{jk}>0$, $\forall i\neq k$, one solution is: $\bZs_j= \overline{\bH}_j(\overline{\bH}_j\herm\overline{\bH}_j)^{-1} {\bf D}_{j}(\overline{\bH}_j\herm\overline{\bH}_j)^{-1}\overline{\bH}_j\herm$, where ${\bf D}_j$ is an arbitrary diagonal matrix which can be used for further optimization of the ergodic rate.} { Choosing $\bZs_j$ as such renders $\bQ_j$ diagonal, and hence, eliminates the {uncorrelated} inter-cell interference when using S-MMSE detection.}
\par From Theorems \ref{th:MRC} and \ref{th:S-MMSE}, we deduce that for MRC and S-MMSE receivers, pilot contamination remains a tenacious limitation for massive MIMO systems with Rician fading channels, even under favorable propagation conditions. Additionally, although LoS components clearly enhance the received signal, they also cause interference. Normally, this latter lessens in favorable propagation conditions; however such environments are not constantly available. Therefore, more {research should be done to fully leverage the LoS component so it becomes an enabling performance factor under all propagation scenarios. For instance, a statistical processing scheme that exploits the LoS components and statistics of the scattered signals for LoS-prevailing environments is proposed in \cite{Uplink-ikramICC2018}.} 
    
\subsection{M-MMSE}
Prior to deriving a deterministic equivalent for $\gamma\MMMSE_{jk}$, we consider the following assumption \cite[Assumption 5]{Massive-Emil2018}:
\begin{assumption}\label{ass:LI matrices} $\forall j, \ell, k,$ with $\boldsymbol{\lambda}_{jk}=[\lambda_{j1}k,\dots,\lambda_{jLk}]^{T}$ and $\ell'=1,\dots,L$:
\begin{equation}
\liminf_N \underset{\{\boldsymbol{\lambda}_{jk}:\lambda_{j\ell'k=1}\}}{\inf}\frac{1}{N}\left\Vert\sum_{\ell=1}^{L}\lambda_{j\ell k}\bR_{j\ell k}\right\Vert_{F}^{2}>0.
\end{equation}
\end{assumption}
\newpage
\begin{theorem}\label{th:M-MMSE} Under assumptions \ref{ass:asymptotic} and \ref{ass:LI matrices}, we have:\\ 
\par $\frac{1}{N} \gamma\MMMSE_{jk} - \frac{1}{N} \overline{\gamma}\MMMSE_{jk} \asto 0  \text{, s.t:}$
\begin{equation}\label{eq:gamma MMMSE DE}
\frac{\overline{\gamma}\MMMSE_{jk}}{N} = \frac{1}{\left[\bT_{jk}^{-1}\right]_{11}} + \frac{1}{N} \overline{\bh}_{jk}\herm  \overline{\bQ}_{j,/k} \overline{\bh}_{jk}. %\beta^{\mbox{\tiny M}}_{jk,jj} - \bb_{jk}\herm \bC_{jk}^{-1} \bb_{jk}
\end{equation}
with the $L \times L$ block matrix $\bT_{jk}=\left[\begin{array}{cc}
\beta^{\mbox{\tiny M}}_{jk,jj} & \bb_{jk}\herm \\
\bb_{jk} &  \bC_{jk} \\ \end{array}\right],$ such that:
\begin{itemize}
\item $\beta^{\mbox{\tiny M}}_{jk,nm}= \frac{1}{N}\tr\left(\bR_{jnk} \bPhi_{jk} \bR_{jmk} \bZm_j\right)$,
\item $\bb_{jk}=\left[\beta^{\mbox{\tiny M}}_{jk,1 j} \dots \beta^{\mbox{\tiny M}}_{jk,(j-1) j},\beta^{\mbox{\tiny M}}_{jk,(j+1) j},\dots,\beta^{\mbox{\tiny M}}_{jk,L j}\right]$, %is the $(L-1)$ vector, whose $\ell-$th element is $\left[\bb_{jk}\right]_\ell = \beta^{\mbox{\tiny M}}_{jk,\ell j}$, $\ell\neq j$, 
\item $\bC_{jk}$ $\in\mathbb{R}^{(L-1)^2}$ with entries: $\left[\bC_{jk}\right]_{mn}= \beta^{\mbox{\tiny M}}_{jk,nm}, \ n\neq j \text{ and } m\neq j,$
\end{itemize}
and $\overline{\bQ}_{j,/k}= \left(\frac{1}{N}\overline{\bH}_{j,/k} {\bf D}^{-1}_{jk} \overline{\bH}_{j,/k}\herm+(\bZm_j)^{-1} \right)^{-1},$
where $\mathcal{B}\left({\bf D}_{jk},u,v\right)=diag\displaystyle{\left\{\beta^{\mbox{\tiny M}}_{jm,uv}\right\}_{\substack{m=1 \\ m\neq k}}^{K}}$, and the LoS matrix $\overline{\bH}_{j,/k}=$

$\left[\boldsymbol{0}_{(j-1)(K-1)} \ \overline{\bh}_{j,1} \dots \overline{\bh}_{j,k-1} \overline{\bh}_{j,k+1} \dots \overline{\bh}_{jK} \ \boldsymbol{0}_{(L-j)(K-1)}\right]$. %is the LoS matrix $\overline{\bH}_j$ %\overline{\bh}_{j,k-1} \overline{\bh}_{j,k+1} \dots \overline{\bh}_{jK}\right]$ . \\
\end{theorem}
\begin{proof}
A sketch of the proof is given in Appendix \ref{app:MMMSE DE}. 
\end{proof}
\par First, Assumption \ref{ass:LI matrices} implies that for every UE $k$ in cell $j$, the correlation matrices $\bR_{j\ell k}$, $\ell=1,\dots,L$, are asymptotically linearly independent. %{\color{blue}(refer to \cite{Massive-Emil2018} for more details)}. 
If this holds, as $ N \rightarrow \infty $, the invertibility of $\bT_{jk}$ is verified, as shown in \cite[Appendix G]{Massive-Emil2018}. Second, note that the provided expression in Theorem \ref{th:M-MMSE} is an approximation of the `$\frac{1}{N}-$scaled' SINR, $\frac{1}{N}\gamma\MMMSE$. In this line, Theorem \ref{th:M-MMSE} reveals that, under Assumptions \ref{ass:asymptotic} and \ref{ass:LI matrices}, using M-MMSE gives rise to an SINR (and by the same token a capacity) that grows unboundedly as $N\rightarrow \infty$. This outcome was demonstrated in \cite{Massive-Emil2018} for the correlated Rayleigh fading, and is validated here by Theorem \ref{th:M-MMSE} for correlated Rician fading systems. 
{Moreover, the contribution of the specular signals is epitomized by the term $\frac{1}{N} \overline{\bh}_{jk}\herm  \overline{\bQ}_{j,/k} \overline{\bh}_{jk}$, and as shall be seen with simulation results, stronger LoS components improve the performances generated by M-MMSE.} Accordingly, M-MMSE outperforms all the single-cell detectors including MRC and S-MMSE, whose performances remain limited by pilot contamination. This is due to the structure of $\bg\MMMSE_{jk}$ that involves estimates of all interfering channels $\bh_{j\ell i}$, $\forall(\ell,i)$, which mitigates both intra and inter-cell interference, and yields unbounded spectral efficiency. 

\section{Numerical Results}
\begin{figure}[h]
\centering
{\includegraphics[width=35mm, height=35mm]{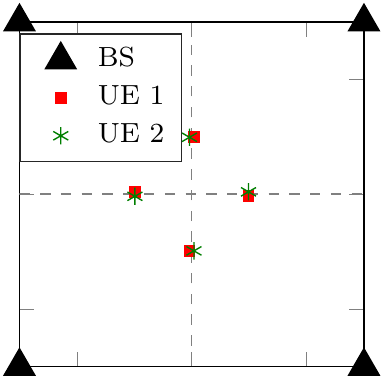}}
\vspace{1em}
\caption{\small Network setup with L=4 cells and K=2 cell-edge users, similarly to \cite{Massive-Emil2018}. The UEs with the same marker use the same pilot.}
\vspace{-1.0 em}
\label{fig:netSEt}
\end{figure}

In this section, we carry out MonteCarlo simulations over $1000$ channel realizations to validate, for finite system dimensions, the asymptotic results provided in Section \ref{sec:Analysis}. To this end, we consider a multi-cell massive MIMO with $L=4$ cells and inner cell-radius of $150$m. Each cell has $K=2$ cell-edge users with similar distance and angle of arrival $\theta$, to ensure high levels of pilot contamination, as depicted in Fig.\ref{fig:netSEt}. We fix $Tc=200$ symbols and $\tau=K$. 
Furthermore, the SNR is the same during training and data transmission (\textit{i.e.} $\beta_{j\ell k}\rho_d=\beta_{j\ell k}\rho_{tr}$ ). It is fixed at $-6$dB for intra-cell UEs and between $-6.3$dB and $-11.5 $dB for the interfering channels from other cells. Additionally, for intra-cell channels, the specular component $\overline{\bz}_{jk}$ follows the model $\left[\overline{\bf z}_{jk}\right]_ n= e^{-j(n-1)\pi\sin\left(\theta_{jjk}\right)}$. Finally, the results are represented in terms of UL rate \eqref{eq:sum-rate} in two scenarios.  
\paragraph*{Scenario I}
In the first setting, we illustrate the performances as a function of the number of  antennas, with different values of the Rician factors $\kappa_{jk}$, and considering the exponential correlation model\cite{Channel-Loyka2001}, such that the elements of the correlation matrix $\boldsymbol{\Theta}_{j\ell k}$ of channel $\bh_{j\ell k}$ are given by : 
\begin{equation}
\left[\boldsymbol{\Theta}_{j\ell k}\right]_{mn} = r^{\vert m -n\vert } e^{\left(i(m-n)\theta_{j\ell k}\right)},
\end{equation}
where $r=0.5$. Accordingly, Fig.\ref{fig:UL rate diff kappaz} displays the corresponding UL achievable rate obtained by MRC, S-MMSE and M-MMSE. Solid and dashed lines depict empirical and asymptotic results, respectively.
\begin{figure}[h]
\centering
\includegraphics[width=\width, height=\height,]{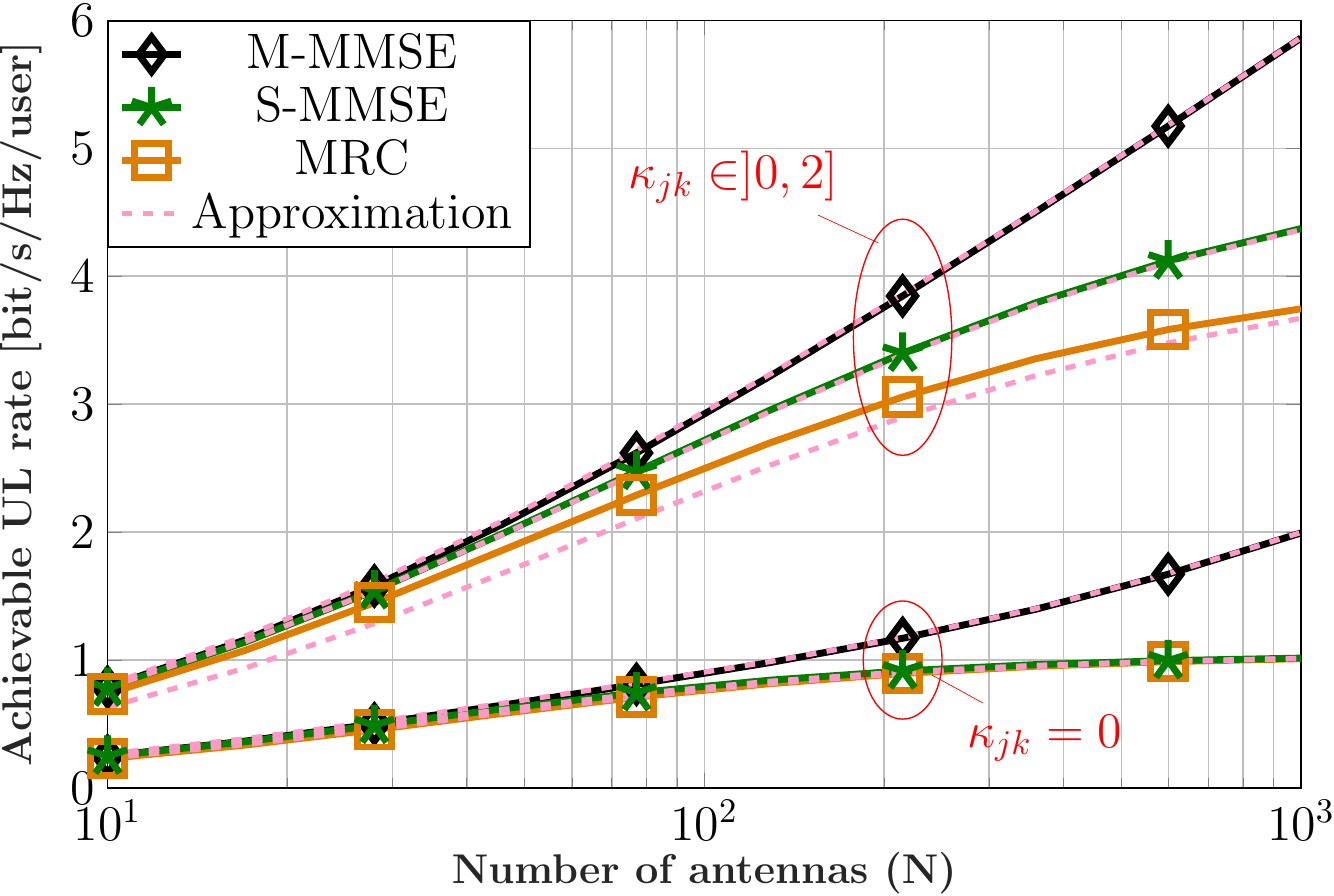}
\vspace{1em}
\caption{\small UL rate using MRC, S-MMSE and M-MMSE in Rayleigh ($\kappa_{jk}=0$) and Rician fading with $\kappa_{jk}\in]0,2]$ for a different number of antennas N, using the exponential channel correlation model with $r=0.5$. Solid and dashed lines depict empirical and asymptotic results, respectively.}
%\vspace{-1.0 em}
\label{fig:UL rate diff kappaz}
\end{figure}
 As can be seen in Fig.\ref{fig:UL rate diff kappaz}, for the three receivers, a better ergodic capacity is obtained in the presence of the LoS (curves with $ 0< \kappa_{jk}\leq 2$), relatively to the case with only scattered signals ($\kappa_{jk}=0$). Plus, S-MMSE provides increasingly better performances than MRC as $\kappa_{jk}$ takes higher values. However, both reach a plateau as $N$ grows large, which is due to pilot contamination, as demonstrated in Theorems \ref{th:MRC} and \ref{th:S-MMSE}. On the other hand, M-MMSE outperforms both single-cell combining schemes, MRC and S-MMSE, and clearly scales linearly with the number of antennas, thus confirming the results of Theorem \ref{th:M-MMSE}. 

\paragraph*{Scenario II}
We now move on to illustrating the performances of the system with the uncorrelated fading model with independent log-normal large-scale variations over the array, such that:
  \begin{equation}
\boldsymbol{\Theta}_{j\ell k} = diag\left\{ 10^{f_i/10}\right\}_{i=1}^{N}, \ \text{ where  } f_i\sim \mathcal{N}(0,\sigma_c^2).
  \end{equation}
To this end, for $N=200$ antennas and different values of $\kappa_{jk}$, we represent in Fig.\ref{fig:diagonal corr} the achievable UL rate with respect to the standard deviation of the fading variations over the array, $\sigma_c$. As can be observed, S-MMSE and MRC generate comparable rates for all values of $\kappa_{jk}$. Furthermore, M-MMSE yields a significantly larger capacity than with MRC and S-MMSE, { except for the special case at ($\sigma_c=0$) corresponding to Rayleigh fading ($\kappa_{jk}=0$), wherein the correlation matrices become linearly dependent, as previously observed in \cite{Massive-Emil2018}. Nevertheless, note that in this particular setting ($\sigma_c=0$), for Rician fading (curves with $\kappa_{jk}\in ]0,2]$), M-MMSE still outperforms both single-cell combining schemes.} 
Finally, Fig.\ref{fig:UL rate diff kappaz} and Fig.\ref{fig:diagonal corr} clearly validate the accuracy of the asymptotic approximations provided in Section \ref{sec:Analysis} (dashed-lines curves).  
\begin{figure}
\centering
\includegraphics[width=\width, height=\height,]{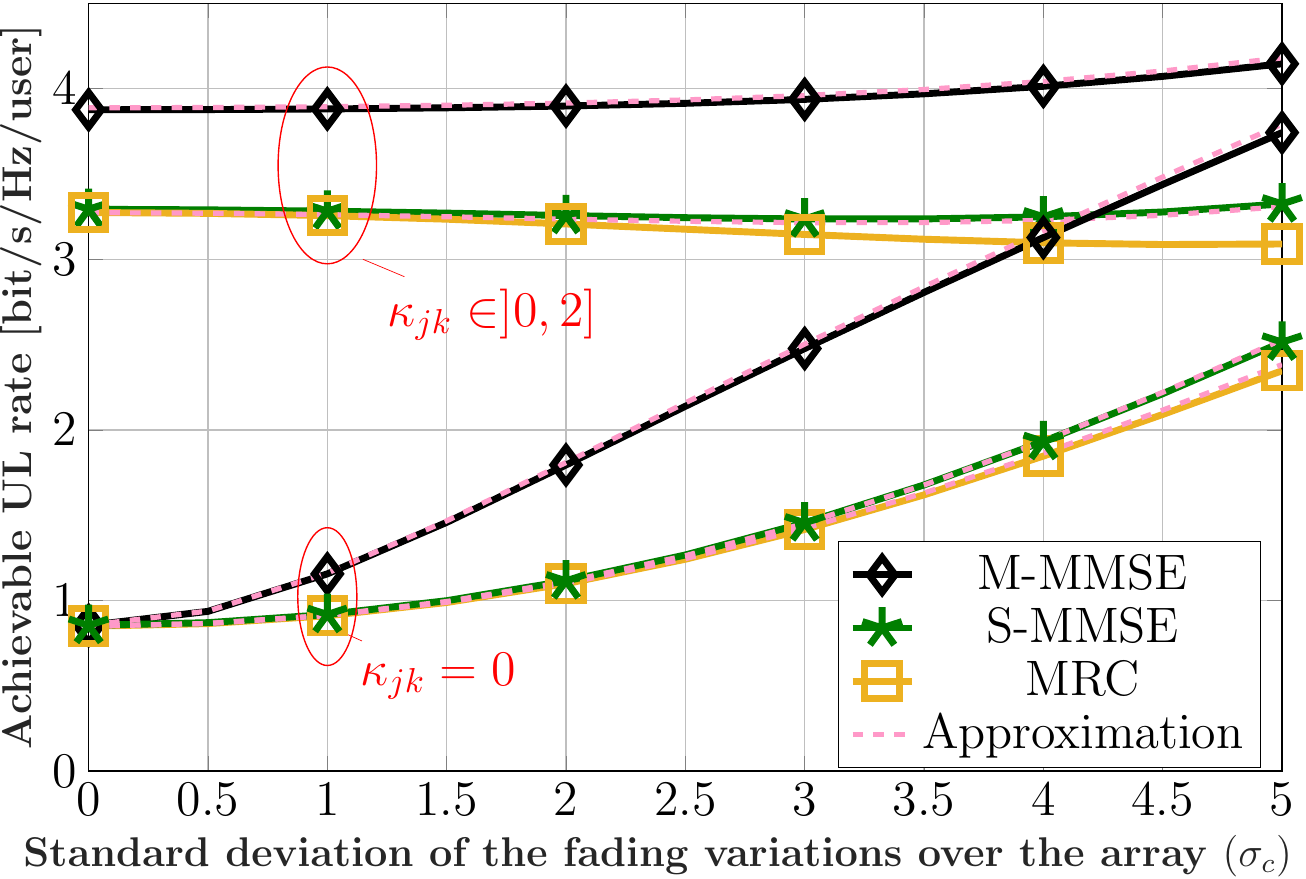}
\vspace{1em}
\caption{\small UL rate using MRC, S-MMSE and M-MMSE in Rayleigh ($\kappa_{jk}=0$) and Rician fading with $\kappa_{jk}\in]0,2]$ for different large-scale fading variations over the array, with $N=200$. Solid and dashed lines depict empirical and asymptotic results, respectively.}
%\vspace{-1.2 em}
\label{fig:diagonal corr}
\end{figure}
\section{Conclusion}
We studied in this work the UL performances of multi-cell massive MIMO systems with correlated Rician fading channels, under the assumption of imperfect channel estimates. Considering the large-antenna limit, we derived closed-form approximations of the achievable rates for different single-cell and multi-cell combining schemes. For the single-cell detection case, we analyzed the MRC and S-MMSE receivers that were shown to produce higher gains as the LoS signals became stronger; yet remain limited by LoS induced interference and pilot contamination as the number of antennas grows large. In contrast, for multi-cell combining, we analytically demonstrated that M-MMSE outperforms any single-cell combining technique. In fact, it provides unlimited capacity that scales linearly with the number of antennas, and is further enhanced in the presence of LoS. A selection of numerical results validates our analysis as it shows that the derived approximations are accurate for finite system dimensions, albeit computed assuming the asymptotic antenna-limit. Plus, these expressions are close to realistic systems since per-user channel correlation and Rician-factor are considered. As such, they provide a theoretical framework that can be harnessed to perform further analysis of similar networks.

%\vspace{-0.2em}
\begin{appendices}
\begin{figure*}[h!]
\begin{equation}\label{eq:gamma M-MMSE expanded}
\frac{\gamma\MMMSE_{jk} }{N} = \underbrace{\frac{1}{N}  \wbhh_{jjk} {\bA}_{j,/k}^{-1} \wbh_{jjk}}_{{\Pi}_1} - \underbrace{\frac{1}{N} \wbhh_{jjk} {\bA}_{j,/k}^{-1} \wbH_{jk,/j}}_{{\Pi}_2} \underbrace{\left(\frac{1}{N}\bI_{L-1} + \frac{1}{N} \wbHh_{jk,/j} \bA_{j,/k}^{-1}\wbH_{jk,/j} \right)^{-1}}_{{\Pi}_3}  \underbrace{ \frac{1}{N} \wbHh_{jk,/j} \bA_{j,/k}^{-1} \wbh_{jjk}}_{{\Pi}\herm_2}.
\end{equation}
\vspace{-1em}
\hrulefill
\end{figure*}
%\vspace{-1em}
\section{Sketch of Proof of Theorems \ref{th:MRC} and \ref{th:S-MMSE}}\label{app:MRC and S-MMSE sketch}
For both Theorems \ref{th:MRC} and \ref{th:S-MMSE}, the derivations rely most often on the same arguments. Therefore, due to space limitations, we mention in this appendix the pertinent steps to derive the asymptotic approximation $\overline{\gamma}\LMMSE_{jk}$ \eqref{eq:gamma MMSE DE simplified}, since this latter is more involved than $\overline{\gamma}\MRC_{jk}$ \eqref{eq:gamma MRC DE}. Therefore, in the sequel of this appendix, $\bg_{jk}$ refers to the S-MMSE receiver \eqref{eq:G_MMSE_hat}.

{\footnotesize \textbullet\ }First we need the following preliminary result : 
Define $\tilde{\bQ}_j=\left(\frac{1}{N}\wbH\herm_{jj} \bZs_j \wbH_{jj}+ \frac{1}{N}\bI_K\right)^{-1}$. Under assumption \ref{ass:asymptotic}, the LLN allows us to put $\frac{1}{N} \left[\wbH\herm_{jj} \bZs_j \wbH_{jj}\right]_{ik}-\frac{1}{N}\mathbb{E} \left[\wbhh_{jji} \bZs_j \wbh_{jjk}\right] \asto 0. 
$ Therefore, using the continuous mapping theorem \cite{book-Prob95}, we have :
% \begin{equation}\label{eq:QtildeQ}
$ [\tilde{\bQ}_j]_{ik} - \left[{\bQ}_j\right]_{ik} \asto 0$,
% \end{equation} 
where the matrix $\bQ_j$ is given in \eqref{eq:Q_hat}.
%\end{equation}  
\par Now, applying the Woodbury identity on the signal term yields : %\begin{equation}\label{eq:sig simp}
$\left\vert \bg\herm_{jk}\wbh_{jjk}\right\vert^2= \left|1-\frac{1}{N} [\tilde{\bf Q}_j]_{kk}\right|^2$.  Accordingly, considering the above equalities, we can find an asymptotic equivalent of the signal: $\vert{\bf g}_{jk}\herm\wbh_{jjk}\vert^2$ $- \left|1-\frac{1}{N} [\bQ_j]_{kk}\right|^2$ $\asto 0$. Likewise, the same steps allow to derive approximations of the intra-cell interference term, estimation errors term and processed noise. 
As to the inter-cell interference, we follow the same reasoning except that we take into account the correlation between the estimates and the interfering channels that share the same pilot, s.t, $\forall \ell\neq j$ : $\frac{1}{N}\mathbb{E} \left[\wbhh_{jji} \bZs_j \wbh_{j\ell i}\right]= \frac{1}{N}\tr(\bR_{j\ell i}\bPhi_{jk}{\bR}_{jj i}\bZs_j)=\beta_{ji, \ell j}^{\mbox{\tiny S}} $. Finally, based on the continuous mapping Theorem \cite{book-Prob95}, putting all these deterministic equivalents together yields the asymptotic approximation of the SINR given in Theorem \ref{th:S-MMSE}. 
%\vspace{-0.5em}
\section{Sketch of Proof of Theorem  \ref{th:M-MMSE}}\label{app:MMMSE DE}

We summarize in this appendix the main steps to obtain the approximation given in Theorem \ref{th:M-MMSE}.  Define the following quantities:
\begin{flushleft}
{%\small
{\footnotesize \textbullet\ } $\bA_{j,/k}={\sum_{\substack{\ell=1}}^{L}\sum_{\substack{i=1 \\ i\neq k}}^{K}\wbh_{j\ell i}\wbhh_{j\ell i}}+(\bZm_j)^{-1}$
$\phantom{=\bA_jkk} = \wbH_{j,/k}\wbH_{j,/k}\herm+(\bZm_j)^{-1}$ ,}
\end{flushleft}
{\footnotesize \textbullet\ }$\wbH_{jk,/j} \in \mathbb{C}^{N\times (L-1)}$ : matrix that contains all the estimates $\wbh_{j \ell k}$, $\forall \ell$, $\ell\neq j$, (\textit{i.e.} all pilot contamintors of channel $\bh_{jjk}$. Plus it does not include any LoS components).Therefore, using the Woodburry identity, we can rewrite $\gamma\MMMSE_{jk}$ as \eqref{eq:gamma M-MMSE expanded}, \textit{(top of this page)}. Next, deriving approximations for $\Pi_1,\Pi_2$ and, $\Pi_3$ is, for the most part, fairly similar to the correlated Rayleigh fading treated in \cite{Massive-Emil2018}. Accordingly, due to space limitations, we provide the main differences that rise in correlated Rician fading and refer the reader to \cite[Appendix F]{Massive-Emil2018} for further details about the common steps. 
\par First, let us focus on the term $\Pi_1$. Using the Woodburry identity enables to rewrite {$\frac{1}{N}\bA_{j,/k}^{-1}= \frac{1}{N} \bZm_j - \frac{1}{N} \bZm_j \wbH_{j,/k} \left( \frac{1}{N} \wbH_{j,/k} \herm \bZm_j \wbH_{j,/k} +\frac{1}{N} \bI_{L(K-1)}\right)^{-1} \frac{1}{N}\wbHh_{j,/k}\bZm_j$}. Note that $\wbH_{j,/k}$, includes all the estimates $\wbh_{j\ell i}$, $\forall \ell$, with $i\neq k$. Therefore, it is fully uncorrelated with $\wbh_{jjk}$, and only has Rician components for $\ell =j$. Consequently, recalling the $\bD_{jk}$ and $\overline{\bH}_{j,/k}$ defined in Theorem \ref{th:M-MMSE}, a direct application of the convergence of quadratic forms lemma \cite{eigenvaluesoutside-silverstein2009}, yields:
{%\small
% \vspace{-0.5 em}
\begin{align} 
&\frac{1}{N} \wbhh _{jjk}\bZm_j\wbh_{jjk}-\left\{\beta^{\mbox{\tiny M}}_{jk,jj}+ \frac{1}{N} \overline{\bh}_{jk}\herm \bZm_j \overline{\bh}_{jk}\right\}\asto 0,\label{eq:hZh} \\
&\frac{1}{N} \wbhh _{jjk}\bZm_j\wbH_{j,/k}- \left\{\frac{1}{N} \overline{\bh}_{jk}\herm \bZm_j \overline{\bH}_{j,/k}\right\} \asto 0,\label{eq:hZH}\\
&\frac{1}{N} \wbH_{j,/k} \herm \bZm_j \wbH_{j,/k} -\left\{\bD_{jk}+ \frac{1}{N} \overline{\bH}_{j,/k}\herm \bZm_j \overline{\bH}_{j,/k}\right\}\asto 0 .\label{eq:HZH}
\end{align}}
%\vspace{-0.3em}
Then, putting together the deterministic equivalents in \eqref{eq:hZh}-\eqref{eq:HZH} provides an approximation of $\Pi_1$. On another note, since $\wbH_{jk,/j}$ does not include any Rician components, dealing with $\Pi_2$ and $\Pi_3$, amounts, to the same reasoning as in \cite[Appendix F]{Massive-Emil2018}. Finally, adding the obtained asymptotic approximations leads to the expression given in Theorem \ref{th:M-MMSE}.
\end{appendices}
%\vspace{-0.5 em}
\bibliographystyle{IEEEtran} 	% (uses file "plain.bst")
\bibliography{ref}
\end{document}